\documentclass[twocolumn,english]{IEEEtran}
\usepackage[T1]{fontenc}
\usepackage{color}
\usepackage{babel}
\usepackage{amsmath}
\usepackage{amssymb}
\usepackage{graphicx}
\usepackage[unicode=true,
 bookmarks=true,bookmarksnumbered=true,bookmarksopen=true,bookmarksopenlevel=1,
 breaklinks=false,pdfborder={0 0 0},backref=false,colorlinks=false]
 {hyperref}
\hypersetup{pdftitle={Your Title},
 pdfauthor={Your Name},
 pdfpagelayout=OneColumn, pdfnewwindow=true, pdfstartview=XYZ, plainpages=false}
\usepackage{breakurl}

\makeatletter
\ifCLASSOPTIONcompsoc
\usepackage[caption=false,font=normalsize,labelfont=sf,textfont=sf]{subfig}
\else
\usepackage[caption=false,font=footnotesize]{subfig}
\fi
\usepackage{cite}
\usepackage{amsthm}
\newtheorem{theorem}{\textbf{Theorem}}

\newtheorem{lemma}{\textbf{Lemma}}
\newtheorem{proposition}{\textbf{Proposition}}
\usepackage{algorithm}
\usepackage{amsmath}

\makeatother

\begin{document}

\title{\textcolor{black}{Effect of Densification on Cellular Network Performance
with Bounded Pathloss Model}}

\author{\IEEEauthorblockN{Junyu Liu, Min Sheng, Lei Liu, Jiandong Li}}
\maketitle
\begin{abstract}
In this paper, we investigate how network densification influences
the performance of downlink cellular network in terms of coverage
probability (CP) and area spectral efficiency (ASE). Instead of the
simplified unbounded pathloss model (UPM), we apply a more realistic
bounded pathloss model (BPM) to model the decay of signal power caused
by pathloss. It is shown that network densification indeed degrades
CP when the base station (BS) density $\lambda$ is sufficiently large.
This is inconsistent with the result derived using UPM that CP is
independent of $\lambda$. Moreover, we shed light on the impact of
ultra-dense deployment of BSs on the ASE scaling law. Specifically,
it is proved that the cellular network ASE scales with rate $\lambda e^{-\kappa\lambda}$,
i.e., first increases with $\lambda$ and then diminishes to be zero
as $\lambda$ goes to infinity.
\end{abstract}

\section{Introduction\label{sec:Introduction}}

\IEEEPARstart{D}{ue} to the simplicity and mathematical tractability,
the unbounded pathloss model (UPM) $g\left(d\right)=d^{-\alpha}$%
\footnote{Note that $d$ denotes the distance from the receiver to the intended
transmitter.%
} has been widely applied to characterize channel power gain caused
by pathloss in wireless networks \cite{unbounded_model_ref_1,unbounded_model_ref_2,unbounded_model_ref_3},
especially when transmission distance is large in the rural areas.
One exhilarating result derived using this model is that the area
spectral efficiency (ASE) is monotonically increasing with the base
station (BS) density in heavily loaded cellular networks \cite{unbounded_model_ref_3}.
However, as the network density becomes larger in the future fifth
generation (5G) wireless networks, it becomes more likely that the
transmission distance is small. Despite its simplicity, UPM fails
to accurately characterize channel power gain in this case. In particular,
when $d\in\left(0,1\right)$, applying UPM would artificially force
the received signal power to be greater than the transmitted signal
power, which is physically impossible. 

With this regard, a more realistic model, namely, bounded pathloss
model (BPM), has been adopted to model the channel power gain caused
by pathloss, especially for dense urban scenarios. Widely applied
BPMs include $\left(1+d\right)^{-\alpha}$, $\left(1+d^{\alpha}\right)^{-1}$
and $\min\left(1,d^{-\alpha}\right)$. In literature, the impact of
BPM on wireless network performance has been extensively investigated
\cite{Ref_cluster_BPM_UPM,BPM_original_ref,Ref_applying_BPM}. In
\cite{Ref_cluster_BPM_UPM}, authors have figured out the influence
of UPM and BPM on the performance of clustered wireless ad hoc networks.
To be specific, depending on the user density, it is shown that the
benefits of clustering is greatly overestimated using UPM. Meanwhile,
the results in \cite{BPM_original_ref} indicate that the probability
density function (PDF) of the interference signal strength becomes
heavy-tailed under UPM, while quickly decays to be zero under BPM.
The difference is due to the singularity of the UPM at 0. Accordingly,
compared to BPM, the application of UPM leads to significant deviations
when evaluating the network performance, such as bit error rate and
wireless channel capacity, etc. However, to our best knowledge, the
effect of BPM on how ASE scales with network density in cellular downlink
networks remains to be explored.

In this paper, we investigate the influence of BPM on the key parameters
of cellular networks, i.e., coverage probability (CP) and ASE. It
is shown that CP is invariant of BS density in sparse scenarios, while
is dramatically degraded by the increasing BS density when BSs are
over-deployed. Based on this result, we further prove that the ASE
first increases and then decreases with the BS density under BPM,
which is different from the results in \cite{unbounded_model_ref_3}.
In addition, the optimal BS density, which leads to the largest network
ASE, can be numerically obtained or be approximated in closed-form
according to the analysis. The results are useful for the BS deployment
and network design.

\section{System Model\label{sec:System-Model}}

We consider \textcolor{black}{a cellular network, which consists BSs
and downlink users. Two independent homogeneous Poisson Point Processes
(HPPPs), $\Pi_{\mathrm{BS}}=\left\{ \mathrm{BS}_{i}\right\} $ and
$\Pi_{\mathrm{U}}=\left\{ \mathrm{U}_{j}\right\} $ $\left(i,j\in\mathbb{N}\right)$,
are used to model the locations of  BSs and downlink users, respectively,
in the infinitely large two-dimensional plane. A distance-based association
rule has been adopted, i.e., each cellular user is connected to the
geographically closest BS} \textcolor{black}{with constant transmit
power $P_{\mathrm{BS}}$}.\textcolor{black}{{} Meanwhile, we consider
a heavily loaded network, in which user density is much greater than
the BS density $\lambda$, such that each BS is connected with at
least one user}%
\footnote{\textcolor{black}{Note that each BS serves one use at one time and
users are served in a round robin manner if more than one user is
connected to one BS.}%
}\textcolor{black}{. Besides, BSs are assumed to always have data to
transmit.}

Channel power gain is assumed to consist of a pathloss component and
a distance-independent small-scale fading component. In particular,
to characterize the power gain caused by pathloss, two typical BPMs
are used, i.e., $g_{1}\left(d\right)=\left(1+d\right)^{-\alpha}$
and $g_{2}\left(d\right)=\left(1+d^{\alpha}\right)^{-1}$, where $\alpha>2$
denotes the pathloss exponent. Meanwhile, Rayleigh fading, $H\sim\exp\left(1\right)$,
is used to model the power gain caused by small-scale fading.

\textbf{Notation}: Let $f_{1}\left(x\right)$ and $f_{2}\left(x\right)$
denote two functions defined on the subset of real numbers. Then,
we write $f_{1}\left(x\right)=\Omega\left(f_{2}\left(x\right)\right)$
if $\exists m>0$, $x_{0}$, $\forall x>x_{0}$, $m\left|f_{2}\left(x\right)\right|\leq\left|f_{1}\left(x\right)\right|$
and $f_{1}\left(x\right)=\mathcal{O}\left(f_{2}\left(x\right)\right)$
if $\exists m>0$, $x_{0}$, $\forall x>x_{0}$, $\left|f_{1}\left(x\right)\right|\leq m\left|f_{2}\left(x\right)\right|$.

\section{Coverage Probability Analysis\label{sec:Coverage-Probability-Analysis}}

In this section, we investigate the performance of the downlink cellular
network by evaluating the CP of a typical downlink user $\mathrm{U}_{0}$,
which is defined as
\begin{equation}
\mathsf{P}_{\mathrm{SIR}}\left(\lambda\right)=\mathbb{P}\left(\mathrm{SIR}_{\mathrm{U}_{0}}>\tau\right),\label{eq:definition CP}
\end{equation}
where $\mathrm{SIR}_{\mathrm{U}_{0}}$ denotes the signal-to-interference
ratio%
\footnote{We ignore the impact of thermal noise on network performance, since
noise is negligible in interference-limited networks.%
} (SIR) at $\mathrm{U}_{0}$ and $\tau$ denotes the SIR threshold.
Denoting $d_{i}$ as the distance from $\mathrm{BS}_{i}$ to $\mathrm{U}_{0}$,
$\mathrm{SIR}_{\mathrm{U}_{0}}$ in (\ref{eq:definition CP}) can
be expressed as
\begin{equation}
\mathrm{SIR}_{\mathrm{U}_{0}}=\frac{P_{\mathrm{BS}}g_{n}\left(d_{0}\right)H_{\mathrm{U}_{0},\mathrm{BS}_{0}}}{\underset{\mathrm{BS}_{i}\in\Pi_{\mathrm{BS}}^{\dagger}}{\sum}P_{\mathrm{BS}}g_{n}\left(d_{i}\right)H_{\mathrm{U}_{0},\mathrm{BS}_{i}}},\: n\in\left\{ 1,2\right\} \label{eq:SIR expression}
\end{equation}
where $H_{\mathrm{U}_{0},\mathrm{BS}_{i}}$ denotes the power gain
caused by fading from $\mathrm{BS}_{i}$ to $\mathrm{U}_{0}$ and
$\Pi_{\mathrm{BS}}^{\dagger}=\Pi_{\mathrm{BS}}\backslash\left\{ \mathrm{BS}_{0}\right\} $. 

In the following, we provide the CP of $\mathrm{U}_{0}$ under BPM
in Proposition \ref{proposition: Coverage Probability}. Note that
we denote $HyF_{1}\left(x\right)={}_{2}F_{1}\left(1,1-\delta,2-\delta,-x\right)$
and $HyF_{2}\left(x\right)={}_{2}F_{1}\left(1,1-\frac{\delta}{2},2-\frac{\delta}{2},-x\right)$,
where $\delta=\frac{2}{\alpha}<1$ and $_{2}F_{1}\left(\cdot,\cdot,\cdot,\cdot\right)$
is the Gaussian hypergeometric function, for simplicity throughout
the paper.

\begin{proposition}[CP Under BMP]

Under BPMs $g_{1}\left(d\right)=\left(1+d\right)^{-\alpha}$ and $g_{2}\left(d\right)=\left(1+d^{\alpha}\right)^{-1}$,
the CPs defined by (\ref{eq:definition CP}) are given by (\ref{eq:CP bounded model 1})
and (\ref{eq:CP bounded model 2}), respectively,\begin{small}
\begin{align}
\mathsf{P}_{\mathrm{SIR},g_{1}}\left(\lambda\right) & =\mathbb{E}_{d_{0}}\left[e^{-\pi\lambda\left(1+d_{0}\right)\left(c_{1}\left(1+d_{0}\right)-c_{2}\right)}\right]\nonumber \\
 & =\frac{e^{-\pi\lambda\hat{c}}}{1+c_{1}}+\frac{\pi\sqrt{\lambda}\left(c_{1}+\hat{c}\right)e^{\frac{\pi\lambda\left(c_{2}^{2}-4\hat{c}\right)}{4\left(1+c_{1}\right)}}}{2\left(1+c_{1}\right)^{\frac{3}{2}}}\nonumber \\
 & \times\left(\mathrm{Erfc}\left(\frac{-\sqrt{\pi\lambda}\left(c_{1}+\hat{c}\right)}{2\sqrt{1+c_{1}}}\right)-2\right),\label{eq:CP bounded model 1}
\end{align}

\begin{equation}
\mathsf{P}_{\mathrm{SIR},g_{2}}\left(\lambda\right)=\mathbb{E}_{d_{0}}\left[e^{-\frac{2\pi\lambda\tau\left(1+d_{0}^{\alpha}\right)}{\left(\alpha-2\right)d_{0}^{\alpha-2}}HyF_{1}\left(\frac{1+\tau\left(1+d_{0}^{\alpha}\right)}{d_{0}^{\alpha}}\right)}\right],\label{eq:CP bounded model 2}
\end{equation}
\end{small}where $c_{1}=\frac{2\tau HyF_{1}\left(\tau\right)}{\alpha-2}$,
$c_{2}=\frac{2\tau HyF_{2}\left(\tau\right)}{\alpha-1}$, $\hat{c}=c_{1}-c_{2}$
and $\mathrm{Erfc}\left(\cdot\right)$ denotes the complementary error
function. As each user is associated with the nearest BS, the PDF
of $d_{0}$ is given by $f_{d_{0}}\left(x\right)=2\pi\lambda xe^{-\pi\lambda x^{2}}$,
$x\geq0$.

\label{proposition: Coverage Probability}

\end{proposition}

\begin{proof}According to (\ref{eq:definition CP}) and (\ref{eq:SIR expression}),
we derive the CP as
\begin{align}
\mathsf{P}_{\mathrm{SIR}}\left(\lambda\right)= & \mathbb{P}\left(H_{\mathrm{U}_{0},\mathrm{BS}_{0}}>s\underset{\mathrm{BS}_{i}\in\Pi_{\mathrm{BS}}^{\dagger}}{\sum}H_{\mathrm{U}_{0},\mathrm{BS}_{i}}g_{n}\left(d_{i}\right)\right)\nonumber \\
\overset{\left(\mathrm{a}\right)}{=} & \mathbb{E}_{d_{0},\Pi_{\mathrm{BS}}^{\dagger},H_{\mathrm{U}_{0},\mathrm{BS}_{i}}}\left[\underset{\mathrm{BS}_{i}\in\Pi_{\mathrm{BS}}^{\dagger}}{\prod}e^{-sH_{\mathrm{U}_{0},\mathrm{BS}_{i}}g_{n}\left(d_{i}\right)}\right]\nonumber \\
\overset{\left(\mathrm{b}\right)}{=} & \mathbb{E}_{d_{0},\Pi_{\mathrm{BS}}^{\dagger}}\left[\underset{\mathrm{BS}_{i}\in\Pi_{\mathrm{BS}}^{\dagger}}{\prod}\frac{1}{1+sg_{n}\left(d_{i}\right)}\right]\nonumber \\
\overset{\left(\mathrm{c}\right)}{=} & \mathbb{E}_{d_{0}}\left[e^{-2\pi\lambda\int_{d_{0}}^{\infty}x\left(1-\frac{1}{1+sg_{n}\left(x\right)}\right)dx}\right],\label{eq:CP proof general}
\end{align}
where $s=\frac{\tau}{g_{n}\left(d_{0}\right)}$. In (\ref{eq:CP proof general}),
(a) and (b) follow due to $H_{\mathrm{U}_{0},\mathrm{BS}_{0}}\sim\exp\left(1\right)$
and $H_{\mathrm{U}_{0},\mathrm{BS}_{i}}\sim\exp\left(1\right)$, respectively,
and (c) follows due to the probability generating functional (PGFL)
of Poisson point process (PPP). Replacing $g_{n}\left(d_{i}\right)$
with $g_{1}\left(d_{i}\right)=\left(1+d_{i}\right)^{-\alpha}$ and
$g_{2}\left(d_{i}\right)=\left(1+d_{i}^{\alpha}\right)^{-1}$, respectively,
we complete the proof.\end{proof}

According to Proposition 1, we can numerically obtain the scaling
law of CP. Fig. \ref{fig:CP scaling law} plots the CP as a function
of BS density varying SIR thresholds under pathloss models $g\left(d\right)=d^{-\alpha}$,
$g_{1}\left(d\right)$ and $g_{2}\left(d\right)$, respectively. Note
that the CP derived using $g\left(d\right)$ is obtained from the
results in \cite{unbounded_model_ref_3}. It is observed that the
CPs are $\lambda$-invariant when $\lambda$ is sufficiently small.
The intuition behind this is that the increase of the received signal
power is counter-balanced by the increase of interference power. Hence,
the impact of $\lambda$ on CP is neutralized. Furthermore, we can
see that the gap between the CPs derived using BPMs and that derived
using UPM is small. This indicates that UPM can accurately model the
channel power gain caused by pathloss in sparse networks, where transmission
distance is basically large. Nevertheless, when the network is further
densified, the difference of UPM and BPM in impacting CP variation
behavior becomes evident. Specifically, under BPM, the CP is greatly
reduced with increasing $\lambda$ (e.g., $\lambda\in\left[0.1,1\right]\:\mathrm{BSs}/\mathrm{m}^{2}$
in Fig. \ref{fig:CP scaling law}) and decays to be zero when $\lambda$
is sufficiently large (e.g., $\lambda>1\:\mathrm{BSs}/\mathrm{m}^{2}$
in Fig. \ref{fig:CP scaling law}). The result manifests that user
experience is significantly degraded by the over-deployment of BSs.
In the next section, the influence of network densification on the
network performance, i.e., network ASE, is explored.

\begin{figure}[t]
\begin{centering}
\includegraphics[width=3.5in]{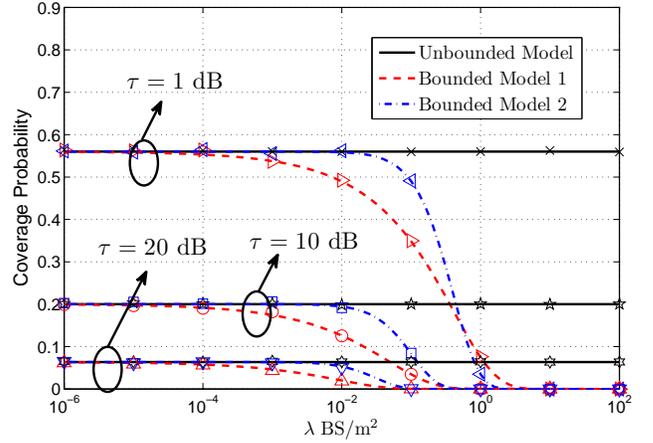}
\par\end{centering}

\caption{\label{fig:CP scaling law}CP scaling with BS density. System parameters
are set as $P_{\mathrm{BS}}=20\:\mathrm{dBmW}$ and $\alpha=4$. Numerical
results and simulation results are drawn by lines and markers, respectively.}
\end{figure}

\section{Scaling Law of Area Spectral Efficiency}

In this section, we study the network ASE and investigate the ASE
scaling law. In particular, the ASE of the downlink cellular network
is defined as
\begin{equation}
\mathcal{A}=\lambda\mathsf{P}_{\mathrm{SIR}}\left(\lambda\right)\mathrm{log}_{2}\left(1+\tau\right).\:\left[\mathrm{bits}/\left(\mathrm{s\cdot Hz\cdot m^{2}}\right)\right]\label{eq:definition ASE}
\end{equation}
By definition, it is easy to derive the network ASE based on Proposition
\ref{proposition: Coverage Probability} when two typical BPMs are
considered. However, since the exact results of the ASE are in complicated
forms, it is difficult to directly observe how ASE scales with the
BS density. To this end, we analyze the scaling law of ASE upper bound
and lower bound. Before providing the bounds of ASE, we first give
the following lemma.

\begin{lemma}

Denote $F_{1}\left(x\right)=HyF_{1}\left(x\right)$ and $F_{2}\left(x\right)=\frac{HyF_{1}\left(x\right)}{\alpha-2}-\frac{HyF_{2}\left(x\right)}{\alpha-1}$
$\left(x\geq0\right)$. Then, $F_{1}\left(x\right)$ and $F_{2}\left(x\right)$
are monotonically decreasing functions of $x$.

\label{lemma: Hypergeometric function feature}

\end{lemma}

\begin{proof} The monotonicity of $F_{1}\left(x\right)$ and $F_{2}\left(x\right)$
can be obtained by showing $\frac{dF_{1}\left(x\right)}{dx}<0$ and
$\frac{dF_{2}\left(x\right)}{dx}<0$. By definition, $\frac{dF_{1}\left(x\right)}{dx}=-\frac{1-\delta}{2-\delta}{}_{2}F_{1}\left(2,2-\delta,3-\delta,-x\right)$.
According to \cite[Theorem 3]{Lemma_ref}, we have $\frac{dF_{1}\left(x\right)}{dx}<-\frac{1-\delta}{2-\delta}\frac{1}{\left(1+x\frac{2-\delta}{3-\delta}\right)^{2}}<0$.
Then, $F_{1}\left(x\right)$ is monotonically decreasing with $x$.
Meanwhile, $\frac{dF_{2}\left(x\right)}{dx}=\frac{HyF_{2}\left(x\right)-HyF_{1}\left(x\right)}{\alpha x}$.
According to \cite[Theorem 1]{Lemma_ref}, $F_{3}\left(x\right)=\frac{HyF_{2}\left(x\right)}{HyF_{1}\left(x\right)}$
is a decreasing function of $x$, the maximum of which equals 1 when
$x=0$. Therefore, we have $HyF_{2}\left(x\right)<HyF_{1}\left(x\right)$
and $\frac{dF_{2}\left(x\right)}{dx}<0$. Hence, we complete the proof.\end{proof}

Based on Lemma \ref{lemma: Hypergeometric function feature}, we study
the upper/lower bounds of ASE under BPMs in the following proposition.

\begin{proposition}[ASE Upper/Lower Bound Under BMPs]

Under BPMs $g_{1}\left(d\right)=\left(1+d\right)^{-\alpha}$ and $g_{2}\left(d\right)=\left(1+d^{\alpha}\right)^{-1}$,
the network ASE is upper bounded by $\mathcal{A}^{\mathrm{U}}\left(\lambda\right)=\lambda\mathrm{log}_{2}\left(1+\tau\right)\frac{e^{-\pi\lambda2^{-\alpha}\hat{c}}}{1+2^{-\alpha}\hat{c}}$
and lower bounded by $\mathcal{A}^{\mathrm{L}}\left(\lambda\right)$
\begin{eqnarray}
 & = & \lambda\mathrm{log}_{2}\left(1+\tau\right)\left[\frac{e^{-\pi\lambda\left(1+2^{\alpha}c_{1}\right)}}{1+2^{\alpha-2}c_{1}}\right.\nonumber \\
 & - & \left.\frac{2^{\alpha-2}c_{1}\pi\sqrt{\lambda}e^{-\frac{\pi\lambda c_{1}2^{\alpha-2}}{1+2^{\alpha-2}c_{1}}}\mathrm{Erfc}\left(\frac{\sqrt{\pi\lambda}\left(1+2^{\alpha-1}c_{1}\right)}{\sqrt{1+2^{\alpha-2}c_{1}}}\right)}{\left(1+2^{\alpha-2}c_{1}\right)^{\frac{3}{2}}}\right].\label{eq:ASE lower bound}
\end{eqnarray}

\label{proposition: bounds using bounded model}

\end{proposition}

\begin{proof}When $g_{n}\left(d\right)$ is replaced by $g_{1}\left(d\right)=\left(1+d\right)^{-\alpha}$,
we obtain the lower bound of CP based on (\ref{eq:CP bounded model 1})
as
\begin{align}
\mathsf{P}_{\mathrm{SIR},g_{1}}\left(\lambda\right)\overset{\left(\mathrm{a}\right)}{>} & \mathbb{E}_{d_{0}}\left[e^{-\pi\lambda c_{1}\left(1+d_{0}\right)^{2}}\right]\nonumber \\
= & \frac{e^{-\pi\lambda c_{1}}}{1+c_{1}}-\frac{e^{-\frac{\pi\lambda c_{1}}{1+c_{1}}}\pi\sqrt{\lambda}c_{1}}{\left(1+c_{1}\right)^{\frac{3}{2}}}\mathrm{Erfc}\left(\frac{\sqrt{\pi\lambda}c_{1}}{\sqrt{1+c_{1}}}\right)\nonumber \\
= & \mathsf{P}_{\mathrm{SIR},g_{1}}^{\mathrm{L}}\left(\lambda\right),\label{eq:CP lower bound model 1 proof}
\end{align}
where (a) follows due to the fact that $c_{1}\left(1+d_{0}\right)-c_{2}<c_{1}\left(1+d_{0}\right)$
since $c_{2}>0$, and $e^{-x}$ is a decreasing function of $x$.
Next, we obtain the upper bound of CP as 
\begin{align}
\mathsf{P}_{\mathrm{SIR},g_{1}}\left(\lambda\right)\overset{\left(\mathrm{a}\right)}{<} & \mathbb{E}_{d_{0}}\left[e^{-\pi\lambda\hat{c}\left(1+d_{0}\right)^{2}}\right]\nonumber \\
= & \frac{e^{-\pi\lambda\hat{c}}}{1+\hat{c}}-\frac{e^{-\frac{\pi\lambda\hat{c}}{1+\hat{c}}}\pi\sqrt{\lambda}\hat{c}}{\left(1+\hat{c}\right)^{\frac{3}{2}}}\mathrm{Erfc}\left(\frac{\sqrt{\pi\lambda}\hat{c}}{\sqrt{1+\hat{c}}}\right)\nonumber \\
= & \mathsf{P}_{\mathrm{SIR},g_{1}}^{\mathrm{U}}\left(\lambda\right),\label{eq:CP upper bound model 1 proof}
\end{align}
where (a) follows because $c_{1}\left(1+d_{0}\right)-c_{2}>c_{1}\left(1+d_{0}\right)-c_{2}\left(1+d_{0}\right)$
in (\ref{eq:CP bounded model 1}) and $e^{-x}$ is a decreasing function
of $x$.

When $g_{n}\left(d\right)$ is replaced by $g_{2}\left(d\right)=\left(1+d^{\alpha}\right)^{-1}$,
the CP in (\ref{eq:definition CP}) turns into
\begin{align*}
\mathsf{P}_{\mathrm{SIR},g_{2}}\left(\lambda_{\mathrm{BS}}\right)= & \mathbb{P}\left(\frac{P_{\mathrm{BS}}H_{\mathrm{U}_{0},\mathrm{BS}_{0}}g_{2}\left(d_{0}\right)}{\underset{\mathrm{BS}_{i}\in\Pi_{\mathrm{BS}}^{\dagger}}{\sum}P_{\mathrm{BS}}H_{\mathrm{U}_{0},\mathrm{BS}_{i}}g_{2}\left(d_{i}\right)}>\tau\right).
\end{align*}
Since $g_{1}\left(d\right)<g_{2}\left(d\right)$, we derive the lower
bound of CP by weakening the useful signal power received at $\mathrm{BS}_{0}$
using the BPM $g_{1}\left(d_{0}\right)$. Accordingly, we have
\begin{align}
\mathsf{P}_{\mathrm{SIR},g_{2}}\left(\lambda\right) & >\mathbb{P}\left(\frac{P_{\mathrm{BS}}H_{\mathrm{U}_{0},\mathrm{BS}_{0}}g_{1}\left(d_{0}\right)}{\underset{\mathrm{BS}_{i}\in\Pi_{\mathrm{BS}}^{\dagger}}{\sum}P_{\mathrm{BS}}H_{\mathrm{U}_{0},\mathrm{BS}_{i}}g_{2}\left(d_{i}\right)}>\tau\right)\nonumber \\
 & \overset{\left(\mathrm{a}\right)}{=}\mathbb{E}_{d_{0}}\left[e^{-\frac{2\pi\lambda\tau\left(1+d_{0}\right)^{\alpha}}{\left(\alpha-2\right)d_{0}^{\alpha-2}}HyF_{1}\left(\frac{1+\tau\left(1+d_{0}^{\alpha}\right)}{d_{0}^{\alpha}}\right)}\right]\nonumber \\
 & \overset{\left(\mathrm{b}\right)}{>}\mathbb{E}_{d_{0}}\left[e^{-\frac{\pi\lambda c_{1}\left(1+d_{0}\right)^{\alpha}}{d_{0}^{\alpha-2}}}\right]\nonumber \\
 & >\mathbb{E}_{d_{0}}\left[e^{-\frac{\pi\lambda c_{1}\left(1+d_{0}\right)^{\alpha}}{d_{0}^{\alpha-2}}}\left|d_{0}\in\left[1,\infty\right)\right.\right]\nonumber \\
 & \overset{\left(\mathrm{c}\right)}{\geq}\mathbb{E}_{d_{0}}\left[e^{-\pi\lambda c_{1}\left(1+d_{0}\right)^{\alpha}\left(\frac{1+d_{0}}{2}\right)^{2-\alpha}}\left|d_{0}\in\left[1,\infty\right)\right.\right]\nonumber \\
 & =\mathsf{P}_{\mathrm{SIR},g_{2}}^{\mathrm{L}}\left(\lambda\right)<\mathsf{P}_{\mathrm{SIR},g_{1}}^{\mathrm{L}}\left(\lambda\right).\label{eq:CP lower bound model 2 proof}
\end{align}
In (\ref{eq:CP lower bound model 2 proof}), the derivation step of
(a) is similar to those in (\ref{eq:CP proof general}), (b) follows
because $F_{1}\left(x\right)=HyF_{1}\left(x\right)$ is a decreasing
function of $x$ according to Lemma \ref{lemma: Hypergeometric function feature}
and (c) follows because $d_{0}^{2-\alpha}\leq\left(\frac{1+d_{0}}{2}\right)^{2-\alpha}$
when $d_{0}\in\left[1,\infty\right)$. Similarly, we weaken the interference
signal power received at $\mathrm{BS}_{0}$ by replacing $g_{2}\left(d_{i}\right)$
with $g_{1}\left(d_{i}\right)$. Hence, the CP upper bound can be
obtained as follows
\begin{align}
 & \mathsf{P}_{\mathrm{SIR},g_{2}}\left(\lambda\right)\nonumber \\
< & \mathbb{P}\left(\frac{P_{\mathrm{BS}}H_{\mathrm{U}_{0},\mathrm{BS}_{0}}g_{2}\left(d_{0}\right)}{\underset{\mathrm{BS}_{i}\in\Pi_{\mathrm{BS}}^{\dagger}}{\sum}P_{\mathrm{BS}}H_{\mathrm{U}_{0},\mathrm{BS}_{i}}g_{1}\left(d_{i}\right)}>\tau\right)\nonumber \\
= & \mathbb{E}_{d_{0}}\left[e^{-\frac{2\pi\lambda s_{2}}{\left(1+d_{0}\right)^{\alpha-2}}\left(\frac{HyF_{1}\left(\frac{s_{2}}{\left(1+d_{0}\right)^{\alpha}}\right)}{\alpha-2}-\frac{HyF_{2}\left(\frac{s_{2}}{\left(1+d_{0}\right)^{\alpha}}\right)}{\alpha-1}\right)}\right]\nonumber \\
\overset{\left(\mathrm{a}\right)}{<} & \mathbb{E}_{d_{0}}\left[e^{-\pi\lambda\hat{c}\left(1+d_{0}^{\alpha}\right)\left(1+d_{0}\right)^{2-\alpha}}\right]\nonumber \\
\overset{\left(\mathrm{b}\right)}{<} & \mathbb{E}_{d_{0}}\left[e^{-\pi\lambda\hat{c}\frac{\left(1+d_{0}\right)^{\alpha}}{2^{\alpha}}\left(1+d_{0}\right)^{2-\alpha}}\right]\nonumber \\
< & \mathbb{E}_{d_{0}}\left[e^{-\pi\lambda\hat{c}\frac{1+d_{0}^{2}}{2^{\alpha}}}\right]\nonumber \\
= & \mathsf{P}_{\mathrm{SIR},g_{2}}^{\mathrm{U}}\left(\lambda\right)>\mathsf{P}_{\mathrm{SIR},g_{1}}^{\mathrm{U}}\left(\lambda\right),\label{eq:CP upper bound model 2 proof}
\end{align}
where (a) follows because $F_{2}\left(x\right)=\frac{HyF_{1}\left(x\right)}{\alpha-2}-\frac{HyF_{2}\left(x\right)}{\alpha-1}$
is a decreasing function of $x$ according to Lemma \ref{lemma: Hypergeometric function feature}
and (b) follows due to $1+d_{0}^{\alpha}\geq\left(\frac{1+d_{0}}{2}\right)^{\alpha}$.
Combining the results in (\ref{eq:CP lower bound model 1 proof}),
(\ref{eq:CP upper bound model 1 proof}), (\ref{eq:CP lower bound model 2 proof})
and (\ref{eq:CP upper bound model 2 proof}), we complete the proof.\end{proof}

Based on Proposition \ref{proposition: bounds using bounded model},
we characterize the scaling law of the ASE using the following theorem.

\begin{theorem}[ASE Scaling Under BMP]

Under the BPMs $g_{1}\left(d\right)=\left(1+d\right)^{-\alpha}$ and
$g_{2}\left(d\right)=\left(1+d^{\alpha}\right)^{-1}$, network ASE
scales with rate $\lambda e^{-\kappa\lambda}$, where $\kappa$ is
a constant.

\label{theorem: ASE scaling law}

\end{theorem}

\begin{proof}The result is obtained by proving that $\mathcal{A}^{\mathrm{L}}\left(\lambda\right)=\Omega\left(\lambda e^{-\pi\lambda\left(1+2^{\alpha}c_{1}\right)}\right)$
and $\mathcal{A}^{\mathrm{U}}\left(\lambda\right)=\mathcal{O}\left(\lambda e^{-\pi\lambda2^{-\alpha}\hat{c}}\right)$.
According to Proposition \ref{proposition: bounds using bounded model},
it is obtained that $\left|\mathcal{A}^{\mathrm{U}}\left(\lambda\right)\right|\leq\lambda\mathrm{log}_{2}\left(1+\tau\right)e^{-\pi\lambda2^{-\alpha}\hat{c}}$.
Therefore, it can be shown that $\mathcal{A}^{\mathrm{U}}\left(\lambda\right)=\mathcal{O}\left(\lambda e^{-\pi\lambda2^{-\alpha}\hat{c}}\right)$.
Then, considering the ASE lower bound, we have $\left|\mathcal{A}^{\mathrm{L}}\left(\lambda\right)\right|\geq\frac{\lambda\mathrm{log}_{2}\left(1+\tau\right)}{1+2^{\alpha-2}c_{1}}\left(Q_{1}\left(\lambda\right)-Q_{2}\left(\lambda\right)\right)$,
where $Q_{1}\left(\lambda\right)=e^{-\pi\lambda\left(1+2^{\alpha}c_{1}\right)}$
and $Q_{2}\left(\lambda\right)=\frac{2^{\alpha-2}\pi c_{1}\sqrt{\lambda}e^{-\frac{\pi\lambda c_{1}2^{\alpha-2}}{1+c_{1}2^{\alpha-2}}}}{\sqrt{1+2^{\alpha-2}c_{1}}}\mathrm{Erfc}\left(\frac{\sqrt{\pi\lambda}\left(c_{1}2^{\alpha-1}+1\right)}{\sqrt{c_{1}2^{\alpha-2}+1}}\right)$.
Through showing $\exists\lambda_{0}>0$, $\forall\lambda>\lambda_{0}$,
$\frac{Q_{2}\left(\lambda\right)}{Q_{1}\left(\lambda\right)}\in\left(0,\frac{1}{2}\right)$,
we have $\left|\mathcal{A}^{\mathrm{L}}\left(\lambda\right)\right|\geq\frac{\lambda\mathrm{log}_{2}\left(1+\tau\right)}{2\left(1+2^{\alpha-2}c_{1}\right)}Q_{1}\left(\lambda\right)$.
Hence, $\forall\lambda>\lambda_{0}$, $\left|\mathcal{A}^{\mathrm{L}}\left(\lambda\right)\right|\geq\frac{\mathrm{log}_{2}\left(1+\tau\right)}{2\left(1+2^{\alpha-2}c_{1}\right)}\left|\lambda e^{-\pi\lambda\left(1+2^{\alpha}c_{1}\right)}\right|$.
Therefore, we have $\mathcal{A}^{\mathrm{L}}\left(\lambda\right)=\Omega\left(\lambda e^{-\pi\lambda\left(1+2^{\alpha}c_{1}\right)}\right)$.\end{proof}

It is shown from Theorem \ref{theorem: ASE scaling law} that the
network ASE first increases and then decreases with $\lambda$. On
the one hand, network densification greatly improves spatial reuse
by reducing the distance between transmitters and receivers. On the
other hand, network over-densification pushes too many interfering
BSs around downlink users, which incurs severe inter-cell interference.
Therefore, when the spatial resources are fully exhausted, the benefits
of spatial reuse vanish and the detriment of network densification
overwhelms, thereby degrading network ASE. More importantly, the results
demonstrate the importance of using BPM to characterize channel power
gain, considering the ultra-dense deployment of BSs.

Fig. \ref{fig:ASE scaling with BS density} plots the ASE scaling
with BS density. It can be seen that the network performance is overestimated
using UPM, as the resulting ASE is always greater than that derived
using BPM. The reason is that the application of UPM artificially
amplifies the useful signal power at $d\in\left(0,1\right)$. Meanwhile,
we observe that the exact results derived using the two BPMs decay
with $\lambda$ at the same rate with those derived using the upper/lower
bound. This indicates the validity of Theorem \ref{theorem: ASE scaling law}.
Additionally, we interestingly find that the optimal $\lambda^{*}$,
which maximizes the system ASE, can be approximated using the densities
$\lambda_{\mathrm{U}}^{*}$ that maximize the ASE upper bound $\mathcal{A}_{g_{1}}^{\mathrm{U}}\left(\lambda\right)=\lambda\mathrm{log}_{2}\left(1+\tau\right)\mathsf{P}_{\mathrm{SIR},g_{1}}^{\mathrm{U}}\left(\lambda\right)$,
where $\mathsf{P}_{\mathrm{SIR},g_{1}}^{\mathrm{U}}\left(\lambda\right)$
is given by (9). Note that $\lambda_{\mathrm{U}}^{*}$ can be derived
in closed-form by solving $\frac{d\mathcal{A}_{g_{1}}^{\mathrm{U}}\left(\lambda\right)}{d\lambda}=0$.
Therefore, the impact of system parameters on $\lambda^{*}$ can be
directly observed, which provides guidance for the deployment of BSs.

\begin{figure}[t]
\begin{centering}
\includegraphics[width=3.5in]{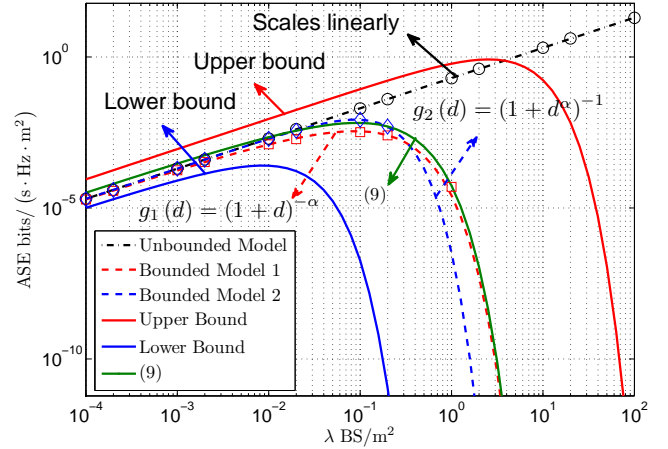}
\par\end{centering}

\caption{\label{fig:ASE scaling with BS density}ASE scaling with BS density.
System parameters are set as $P_{\mathrm{BS}}=20\:\mathrm{dBmW}$,
$\tau=10\:\mathrm{dB}$ and $\alpha=4$. Numerical results and simulation
results are drawn by lines and markers, respectively.}
\end{figure}

\section{Conclusion}

In this paper, we show that the cellular network ASE scales with rate
$\lambda e^{-\kappa\lambda}$ when BPM is used to characterize pathloss.
According to the scaling law, network densification cannot alway boost
the network capacity especially when BS density is sufficiently large.
This differs from the traditional understanding that cellular network
ASE scales linearly with $\lambda$. Meanwhile, the closed-form expression
of the density, which leads to the inflection of the ASE, can be approximated.
The result is helpful to understand how network parameters affect
the network scaling law, thereby providing guideline for the efficient
deployment of cellular networks.

It is worth noting that the single-slope BPM used throughout the paper
cannot characterize the discrepant power decay levels within different
regions, which is caused by non-line-of-sight (NLOS) and line-of-sight
(LOS) propagation of the signal. Recently, the influence of NLOS and
LOS transmissions on dense network performance, e.g., CP and ASE,
has been evaluated and proved to be significant in \cite{Ref_multi_slope_1,Ref_multi_slope_2}.
Nonetheless, the impact of singularity has not been fully explored
therein. Therefore, future study should consider the multiple pathloss
model, which is defined based on BPM. Following this approach, the
influence of network densification on the network capacity would become
valid and convincing.

\bibliographystyle{IEEEtran}
\bibliography{ver6_double_column_for_Arxiv}

\end{document}